\documentclass[final]{dmtcs-episciences}

\usepackage[utf8]{inputenc}
\usepackage{subfigure}

\usepackage[round]{natbib}

\usepackage{geometry,amssymb,latexsym,amsmath}
\usepackage{tikz,ifthen}
\usetikzlibrary{calc}
\usetikzlibrary{decorations.pathmorphing}

\newtheorem{theorem}{Theorem}
\newtheorem{lemma}[theorem]{Lemma}

\newtheorem{observation}[theorem]{Observation}

\newtheorem{definition}[theorem]{Definition}

\newcommand{\efface}[1]{}

\newcommand{\g}{\gamma}				
\newcommand{\gpw}{\gamma_{\rm{P},1}}		
\newcommand{\gpk}{\gamma_{{\rm P},k}}		

\newcommand{\kPDS}{{$k$-PDS}}

\newcommand{\Pt}[2]{\mathcal{P}_{#1}^{#2}}

\newcommand{\grad}{{\rm rad}_{{\rm P},k}}

\newcommand{\cp}{\, \Box \,}

\author{Paul Dorbec\affiliationmark{1,2} \and
Seethu Varghese\affiliationmark{3} \and A. Vijayakumar\affiliationmark{3}}
\title{Heredity for generalized power domination}

\affiliation{Univ. Bordeaux, LaBRI, France\\
CNRS, LaBRI, France\\
Department of Mathematics, Cochin University of Science and Technology, India}

\keywords{power domination, electrical network monitoring, domination, edge critical graphs, propagation radius}

\received{2015-10-2}

\revised{2016-02-22}

\accepted{2016-03-17}

\begin{document}

\publicationdetails{18}{2016}{3}{5}{1290}

\maketitle
\begin{abstract}
In this paper, we study the behaviour of the generalized power domination number of a graph by small changes on the graph, namely edge and vertex deletion and edge contraction. We prove optimal bounds for $\gamma_{{\rm P},k}(G-e)$, $\gamma_{{\rm P},k}(G/e)$ and for $\gamma_{{\rm P},k}(G-v)$ in terms of $\gpk(G)$, and give examples for which these bounds are tight. We characterize all graphs for which $\gamma_{{\rm P},k}(G-e) = \gamma_{{\rm P},k}(G)+1$ for any edge $e$. We also consider the behaviour of the propagation radius of graphs by similar modifications.   
\end{abstract}

\section{Introduction}
\label{sec:in}
Domination is now a well studied graph parameter, and a classical topic in graph theory. To address the problem of monitoring electrical networks with {\em phasor measurement units} (see~\cite{bamiboad-93}), power domination was introduced as a variation of the classical domination (see \cite{hahehehe-02}). The originality of power domination is the introduction of an additional propagation possibility, relative to the possible use of Kirchhoff's laws in an electrical network. From this propagation, a vertex can end up to be monitored even though it is at a large distance from any vertex selected to carry a phasor measurement unit. The original status of this new parameter and its applied motivation makes a subject of increasing interest from the community. 

\smallskip

All graphs $G = (V(G), E(G))$ considered are finite and simple, that is, without multiple edges or loops. The {\em open neighbourhood} of a vertex $v$ of $G$, denoted by $N_{G}(v)$, is the set of vertices adjacent to $v$. The {\em closed neighbourhood} of $v$ is $N_{G}[v] = N_{G}(v)\cup \{v\}$. For a subset $S$ of vertices, the {\em open} (resp. {\em closed}) {\em neighbourhood} $N_{G}(S)$ (resp. $N_{G}[S]$) of $S$ is the union of the open (resp. closed) neighbourhoods of its elements. 
A vertex $v$ in a graph is said to dominate its closed neighbourhood $N_G[v]$.
A subset $S\subseteq V(G)$ of vertices is a dominating set if $N_G[S]=V(G)$, that is if every vertex in the graph is dominated by some vertex of $S$. The minimum size of a dominating set in a graph $G$ is called its domination number, denoted by $\gamma(G)$. 

We now define the generalized version of power domination, the case when $k=1$ coincides with the original power domination.
For $k$-power domination, we define iteratively a set $\Pt {G,k} i (S)$ of vertices monitored by an initial set $S$ (of PMU). The initial set of vertices monitored by $S$ is defined as the set of dominated vertices $\Pt {G,k} 0 (S) = N_{G}[S]$. This step is sometimes called the domination step. Then this set is iteratively extended by including the whole neighbourhood  of all vertices that are monitored and have at most $k$ non-monitored neighbours. This second part is called the {\em propagation rule}. More formally, we define directly the set of monitored vertices for $k$-power domination following the notation of \cite{chdomora-12} :

\begin{definition}[Monitored vertices]
Let $G$ be a graph, $S \subseteq V(G)$ and $k\ge 0$. The sets $\big(\Pt {G,k} i(S) \big)_{i\ge 0}$ of {\em vertices monitored by $S$ at step $i$} are defined as follows:
\begin{quote}
$\Pt {G,k} 0 (S) = N_{G}[S]$ (domination step), and \\ 
$\Pt {G,k} {i+1} (S) = \bigcup \{ N_{G}[v]\colon v\in \Pt {G,k} i (S)$
    such that $\big|N_{G}[v]\setminus \Pt {G,k} i (S) \big| \le k\}$ (propagation steps).
\end{quote}
\end{definition}

Let us make some observations about this definition. First, the set of monitored vertices is monotone by inclusion, \textit{i.e.} $\Pt {G,k} {i} (S) \subseteq \Pt {G,k} {i+1} (S)$. This is easy to check by induction, using the fact that whenever $N_{G}[v]$ has been included in $\Pt {G,k} {i} (S)$, it is included in $\Pt {G,k} {i+1} (S)$. 
This also implies that $\Pt {G,k} {i} (S)$ is always a union of neighbourhoods. Observe also that if for some integer $i_0$,
$\Pt {G,k} {i_{0}}(S)=\Pt {G,k} {i_{0}+1}(S)$, then $\Pt {G,k} {j} (S)=\Pt {G,k} {i_{0}}(S)$ for all $j\geq i_0$. We thus denote this set $\Pt {G,k} {i_{0}}(S)$ by $\Pt {G,k} {\infty}(S)$. When the graph $G$ is clear from the context, we simplify the notation to $\Pt{k}{i}(S)$ and $\Pt{k}
{\infty}(S)$.

\begin{definition}[$k$-power dominating set]
A set $S$ is a $k$-{\em power dominating set} of $G$ (abbreviated \kPDS) if $\Pt {G,k} {\infty}(S)=V(G)$. The least cardinality of such a set is called the $k$-{\em power domination number} of $G$, denoted by $\gpk(G)$. A $\gpk(G)$-{\em set} is a \kPDS\ in $G$ of cardinality $\gpk(G)$. 
\end{definition}   

Observe that $k$-power domination is also a generalization of domination, that we obtain when we set $k=0$. In \cite{chdomora-12}, the authors showed along with some early results about $k$-power domination that some bounds, extremal graphs and properties can be expressed for any $k$, including the case of domination. In \cite{dohelomora-13}, a bound from \cite{zhkach-06} on regular graphs is also generalized to any $k$. 

The computational complexity of the power domination problem was considered in various papers (\cite{aaz-10, aast-09, gunira-08, hahehehe-02}), in which it was proved to be NP-complete on bipartite and chordal graphs as well as for bounded propagation variants.  Linear-time algorithms are known for computing minimum $k$-power dominating sets in trees (\cite{chdomora-12}) and in block graphs~ (\cite{walech-14}). The problem of characterizing the power domination number of a graph is non trivial for simple families of graphs. Early studies try to characterize it for products of paths/grids \cite{dohe-06,domoklsp-08} though do not reach complete characterization in a few cases. Other studies propose closed formulas for the power domination number in hexagonal grids~(see \cite{FSV-11}) or in Sierpi\'{n}ski graphs~(see \cite{dokl-14}). 

In general, it remains difficult to prove lower bounds on the power domination number of a graph. One reason why it is so is that power domination does not behave well when taking subgraphs. In this paper, we explore in detail the behaviour of the power domination number of a graph when small changes are applied to the graph, \textit{e.g.} removing a vertex or an edge, or contracting an edge. (Recall that the graph obtained by {\em contraction} of an edge $e=xy$, denoted by $G/e$, is obtained from $G-e$ by replacing $x$ and $y$ by a new vertex $v_{xy}$ ({\em contracted vertex}) which is adjacent to all vertices in $N_{G-e}(x)\cup N_{G-e}(y)$.)
In particular, we prove in Section~\ref{sec:domnumb} that though the behaviour of the power domination is similar to the domination in the case of the removal of a vertex, the removal of an edge can decrease the power domination number and the contraction of an edge can increase the power domination number, both phenomena that are impossible in usual domination. We characterize the graphs for which the removal of any edge increases the $k$-power domination number.

Another recent but natural question about power domination is related to the propagation radius. 
In a graph, a vertex that is arbitrarily far apart from any vertex in the set $S$ may eventually get monitored by $S$ as in the case of paths. However, in the applied circumstances of the monitoring of an electrical network, applying too many times Kirchhoff's laws successively would induce an unreasonable cumulated margin of error. With this consideration in mind, it is natural to consider power domination with bounded time constraints, as was first studied in \cite{aaz-10}, and then in \cite{lich-14}. Inspired by this study, the $k$-{\em propagation radius} of a graph $G$ was introduced in \cite{dokl-14} as a way to measure the efficiency of a minimum $k$-power dominating set (\kPDS). It gives the minimum number of propagation steps required to monitor the entire graph over all $\gpk(G)$-sets. 

\begin{definition}
The radius of a \kPDS\ $S$ of a graph $G$ is defined by 
\begin{displaymath}
\grad(G,S) = 1 + \min \{i:\ \Pt {G,k} {i}(S) = V(G)\}\,.
\end{displaymath}
The $k$-{\em propagation radius} of a graph $G$ as defined in \cite{dokl-14} can be expressed as
\begin{displaymath}\grad(G) = \min \{\grad(G,S),\ S\ \text{is a}\ k\text{-PDS\ of}\ G,\ |S| = \gpk(G) \}\,.\end{displaymath}
\end{definition}

\smallskip

We finally recall a few graph notations that we use in the following. We denote by $K_n$ the complete graph on $n$ vertices, by $K_{m,n}$ the bipartite complete graph with partite sets of order $m$ and $n$. The path and cycle on $n$ vertices are denoted by $P_n$ and $C_n$, respectively. 
For two graphs $G$ and $H$, $G\cp H$ denotes the Cartesian product of $G$ and $H$, that is the graph with vertex set $V(G)\times V(H)$ and where two vertices $(g,h)$ and $(g',h')$ are adjacent if and only if either $g=g'$ and $hh'\in E(H)$, or $h=h'$ and $gg'\in E(G)$. 

\section{Variations on the power domination number}\label{sec:domnumb}

Before considering the different cases of vertex removal, edge removal and edge contraction, we propose the following technical lemma which should prove useful. It states that if two graphs differ only on parts that are already monitored, then propagation in the not yet monitored parts behave the same. For a graph $G=(V,E)$ and two subsets $X$ and $Y$ of $V$, we denote by $E_{G}(X,Y)$ the set of edges $uv\in E(G)$ such that $u\in X$ and $v\in Y$. Note that if $X\subseteq Y$, $E_{G}(X,Y)$ contains in particular all edges of the induced subgraph $G[X]$ of $G$ on $X$. All along the rest of the paper, $k$ denotes a positive integer. 

\begin{lemma}\label{l:general}
Let $G=(V_G,E_G)$ and $H=(V_H,E_H)$ be two graphs, $S$ a subset of vertices of $G$ and $i$ a non-negative integer. Define $X = V_G\setminus \Pt {G,k} i (S)$ and the subgraph  $G'$ with vertex set $N_G[X]$ and edge set $E_G(X,N_G[X])$. 

Suppose there exists a subset $Y\subseteq V_H$ such that the subgraph $H'=\left(N_H[Y],E_H(Y,N_H[Y])\right)$ is isomorphic to $G'$ with a mapping $\varphi: N_G[X]\rightarrow N_H[Y]$ that maps $X$ precisely to $Y$. Then, if for some {$k$-power dominating set} $T \subseteq V_H$ and some integer $j$, $Y \subseteq V_H\setminus \Pt {H,k} j (T)$, then $S$ is a \kPDS\ of $G$ and $\grad(G,S)\le i-j+\grad(H,T)$.
\end{lemma}

\begin{proof}
For $\ell\ge 0$, denote by $X^\ell$ and $Y^\ell$ respectively the sets $X \cap \Pt {G,k} {i+\ell} (S)$ and $Y \cap \Pt {H,k} {j+\ell} (T)$.
We prove by induction that for all $\ell$, $Y^\ell \subseteq \varphi(X^\ell)$.

 By hypothesis, $X^0 = \emptyset$ and so $\varphi(X^0) = \emptyset = Y^0$, so it holds for $\ell=0$. Now assume that the property is true for some $\ell\ge 0$. Suppose that some vertex $v=\varphi(u)\in N_H[Y]$ satisfies the conditions for propagation in $H$ at step $j+\ell$, \textit{i.e.} $v\in \Pt {H,k} {j+\ell} (T)$ and $|N_H[v]\setminus \Pt {H,k} {j+\ell} (T)|\le k$. We show that $u$ also satisfies the conditions for propagation in $G$. First, remark that $u$ is monitored at step $i+\ell$: indeed, if $u\notin X$, then by definition of $X$, $u \in \Pt {G,k} {i+\ell}$, otherwise if $u \in X$, then $v \in Y \cap \Pt {H,k} {j+\ell} (T) = Y^\ell$, and thus by induction, $u \in X^\ell \subseteq \Pt {G,k} {i+\ell}$. Now consider any neighbour $u'$ of $u$ not yet dominated. Then $u'\in X \setminus X^\ell $ and $\varphi(u')\in Y\setminus Y^\ell$. Moreover, by the isomorphism between $G'$ and $H'$, $\varphi(u')$ is also adjacent to $v$, and was among the at most $k$ non monitored neighbours of $v$ in $H$. Therefore, $u$ has at most $k$ non monitored neighbours in $G$, and also propagates in $G$. Applying this statement to all vertices in $G'$, we infer that $Y^{\ell+1} \subseteq \varphi(X^{\ell+1})$. By induction, this is also true for $\ell = \grad(H,T)-j-1$, and  we deduce that 
 \begin{displaymath}X=\varphi^{-1}(Y) \subseteq \varphi^{-1}\Big( Y^\ell = \big(\Pt {H,k} {\grad(H,T)-1} (T)\cap Y\big)\Big) \subseteq \Big(X^\ell = \big(\Pt {G,k} {\grad(H,T)-j+i-1} (S)\cap X\big)\Big)\,,\end{displaymath}
 and thus that $S$ is a $k$-power dominating set of $G$ and $\grad(G,S)\le i-j+\grad(H,T)$.
\end{proof}

We now use this lemma to state how the \(k\)-power domination number of a graph may change with atomic variations of the graph.

\subsection{Vertex removal}
We denote by \(G-v\) the graph obtained from \(G\) by removing a vertex $v$ and all its incident edges. Similar to what happens for domination (see \cite{hahesl-98}), we have the following:

\begin{theorem}\label{o:propG-v}
Let $G$ be a graph and $v$ be a vertex in $G$. There is no upper bound for $\gpk(G-v)$ in terms of $\gpk(G)$. On the other hand, we have $\gpk(G-v)\ge \gpk(G)-1$. Moreover, if $\gpk(G-v)= \gpk(G)-1$, then $\grad(G)\le \grad(G-v)$.
\end{theorem}

\begin{proof}
We first prove the lower bound, using Lemma~\ref{l:general}. We define $H=G-v$ with the obvious mapping $\varphi$ from $V(G)\setminus {v}$ to $V(H)$.  Let $T$ be a power dominating set of $H=G-v$, that induces the minimum propagation radius. Then for the set $S=T\cup \{v\}$, the conditions of Lemma~\ref{l:general} hold already from $i=0$ and $j=0$ and the bound follows. Moreover, we also get that $\grad(G,S)\le j-i+\grad(H,T) = \grad(G-v)$\,.
For proving there is no upper bound for $\gpk(G-v)$ in terms of $\gpk(G)$,  we can consider the star with $n$ leaves $K_{1,n}$, for which the removal of the central vertex increases the $k$-power domination number from 1 to $n$.
\end{proof}

We now describe examples that tighten the lower bound of the above theorem or illustrate better the absence of upper bound (in particular for graphs that remain connected). A first example for which the tightness of the lower bound can be observed is the $4\times 4$ grid $P_4\cp P_4$, for which we get $\gpw(P_4\cp P_4)=2$ (see \cite{dohe-06}) and $\gpw((P_4\cp P_4)-v)=1$ for any $v$. Simple examples for larger $k$ are the graphs $K_{k+2,k+2}$, for which the removal of any vertex drops the $k$-power domination number from 2 to 1 (those were the only exceptions in \cite{dohelomora-13}), as well as the complete bipartite graph $K_{k+3,k+3}$ minus a perfect matching. 

We now describe infinite families of graphs to illustrate these bounds. The family of graphs $D_{k,n}$ was defined in \cite{chdomora-12}. It is made of $n$ copies of $k+3$-cliques minus an edge, organized into a cycle, and where the end-vertices of the missing edges are linked to the corresponding vertices in the adjacent cliques in the cycle (see Fig.~\ref{fig.G-vinc}). Note that $\gpk(D_{k,n})=n$, as each copy of $K_{k+3}-e$ must contain a vertex of a $k$-power dominating set. Its propagation radius is 1 since $D_{k,n}$ has a dominating set of size $n$. The removal of an end-vertex of the edges linking two cliques (\textit{e.g.} $u$ in Fig.~\ref{fig.G-vinc}) does not change its $k$-power domination number, but the removal of any other vertex (\textit{e.g.} $v$ in Fig.~\ref{fig.G-vinc}) decreases it by one, and increases the propagation radius from 1 to 2. So this forms an infinite family tightening the lower bound for any value of $k$ and $\gpk(G)$.

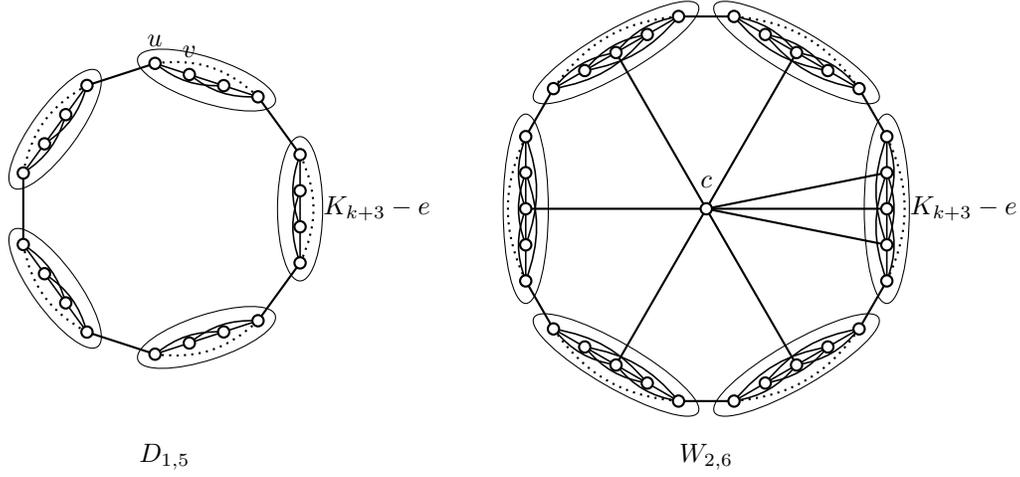
\begin{figure}
\begin{center}
\begin{tikzpicture}[thick,scale=0.6,
	vertex/.style={circle,draw,inner sep=0pt,minimum size=1.5mm,fill=white!100}]

\def\k {3}; 
\def\n {5};


\coordinate (c) at (0,0);

\foreach \i in {1,...,\n}{ 
   \path (\i*360/\n:\k) ++ (90+\i*360/\n:\k/2*0.8) coordinate (x\i0);
   \foreach \j in {1,...,\k}{
     \path (x\i0) ++ (90+\i*360/\n:-0.8*\j) coordinate (x\i\j);
     }
   \draw[thin] (\i*360/\n:\k) circle [x radius=0.5, y radius=\k*0.5+0.1, rotate=\i*360/\n];
   }
 
\draw (x10)--(x2\k) (x20)--(x3\k) (x30)--(x4\k) (x40)--(x5\k) (x50)--(x1\k);

\foreach \i in {1,...,\n}{ 
		\draw[semithick] (x\i0) -- (x\i\k);
		\draw[semithick] (x\i0) to[bend right=20] (x\i2);
		\draw[semithick] (x\i1) to[bend right=20] (x\i3);
		\draw[dotted]  (x\i0) to[bend left=25] (x\i\k);
        }


\foreach \i in {1,...,\n}{ 
	\foreach \j in {0,...,\k}{ 
		\draw (x\i\j) node[vertex] {};
		}
	}

\path (x\n2) ++(1.7,0.4) node {$K_{k+3}-e$};
\path (x10) +(0,0.5) node {$u$};
\path (x11) +(0,0.5) node {$v$};


\begin{scope}[xshift=12cm]

 \def\k {4}; 
 \def\n {6};


 \coordinate (c) at (0,0);

\foreach \i in {1,...,\n}{ 
   \path (\i*360/\n:\k) ++ (90+\i*360/\n:\k/2*0.8) coordinate (x\i0);
   \foreach \j in {1,...,\k}{
     \path (x\i0) ++ (90+\i*360/\n:-0.8*\j) coordinate (x\i\j);
     }
   \draw[thin] (\i*360/\n:\k) circle [x radius=0.5, y radius=\k*0.5+0.1, rotate=\i*360/\n];
   }
 
 \draw (x10)--(x2\k) (x20)--(x3\k) (x30)--(x4\k) (x40)--(x5\k) (x50)--(x6\k) (x60)--(x1\k);

 \foreach \i in {1,...,\n}{ 
		\draw (c) -- (x\i2);
		\draw[semithick] (x\i0) -- (x\i\k);
		\draw[semithick] (x\i0) to[bend left=20] (x\i2);
		\draw[semithick] (x\i0) to[bend right=20] (x\i3);
		\draw[semithick] (x\i1) to[bend left=20] (x\i3);
		\draw[semithick] (x\i1) to[bend right=20] (x\i4);
		\draw[semithick] (x\i2) to[bend left=20] (x\i4);
 		\draw[dotted]  (x\i0) to[bend left=25] (x\i\k);
 }

 \draw (c) -- (x\n1);
 \draw (c) -- (x\n3);
 

 \foreach \i in {1,...,\n}{ 
	\foreach \j in {0,...,\k}{ 
		\draw (x\i\j) node[vertex] {};
		}
	}

 \draw (c) node[vertex] {};
 \path (c) +(0,0.6) node {$c$};
 \path (x\n2) ++(1.7,0) node {$K_{k+3}-e$};

\path (0,-\k-1.5) node {$W_{2,6}$};

\end{scope}

\path (0,-\k-2.5) node {$D_{1,5}$};

\end{tikzpicture}

 \caption{The graphs $D_{k,n}$ and $W_{k,n}$ obtained by the addition of vertex $c$.}
  \label{fig.G-vinc}
\end{center}
\end{figure}

Now, an infinite family of graphs proving the absence of a upper bound is a generalization $W_{k,n}$ of the wheel (depicted in Fig.~\ref{fig.G-vinc}). It is made of $D_{k,n}$ together with a vertex $c$ adjacent to three vertices of degree $k+2$ in one particular clique and to one vertex of degree $k+2$ in all the other cliques. Observe that for $n\ge k+2$, $\{c\}$ is the only power dominating set of $W_{k,n}$ of order 1, and thus we get $\grad(W_{k,n})=\grad(W_{k,n},\{c\})=2+3\lfloor \frac{n-1}{2}\rfloor+2((n-1) \bmod{2})$. The removal of $c$ induces the graph $D_{k,n}$, increasing the $k$-power domination number from 1 to $n$, and dropping the propagation radius from roughly $\frac{3n}{2}$ to 1.

\medskip

More constructions could be proposed to show that the propagation radius of a graph can evolve quite freely when a vertex is removed, and there is little hope for other bounds on this parameter when a vertex is removed. The most unlikely example is that the removal of a vertex increase both the $k$-power domination number and the propagation radius by unbounded value. This is possible with the following variation on $W_{k,pn}$. Consider $pn$ subgraphs $(H_i)_{0\le i<pn}$, all isomorphic to a clique minus an edge, on $k+3$ vertices when $i\equiv 0 \bmod p$ and on $k+1$ vertices otherwise. We again connect the end-vertices of the missing edges in the clique into a cycle joining $H_i$ to $H_{i+1 \pmod{pn}}$, and add a vertex $c$ adjacent to three vertices of degree $k+2$ in all copies $H_i$ when $i\equiv 0 \bmod p$, and to one vertex of degree $k$ in all the other copies. Then $\{c\}$ is a $k$-power dominating set of $G$ inducing a propagation radius of 2. On the other hand, $\gpk(G-c)=n$ (one vertex is needed in each $H_i, i\equiv 0 \bmod p$) and has propagation radius $1+3\lfloor \frac{p-1}{2}\rfloor+2((p-1) \bmod{2})$.

\subsection{Edge removal}

In a graph $G$, removing an edge $e$ can never decrease the domination number. More generally, we have that $\g(G)\le \g(G-e) \le \g(G)+1$. However, the removal of an edge can decrease the $k$-power domination number as stated in the following result. Indeed, it may happen that the removal of one edge allows the propagation through another edge incident to a common vertex, and thus decreases the power domination number.

\begin{theorem}\label{t:propG-e}
Let $G$ be a graph and $e$ be an edge in $G$. Then 
\begin{displaymath}\gpk(G)-1\le \gpk(G-e)\le \gpk(G)+1\,.\end{displaymath}
Moreover,
\begin{displaymath}
\begin{cases}
\mbox{if }\gpk(G)-1 = \gpk(G-e)\mbox{, then } \grad(G)\le \grad(G-e)\\
\mbox{if }\gpk(G-e) = \gpk(G)+1\mbox{, then } \grad(G-e)\le \grad(G).
\end{cases}\end{displaymath}

\end{theorem}

\begin{proof}
We first prove that $\gpk(G-e)\le \gpk(G)+1$. Let $T$ be a $\gpk(G)$-set. If $T$ is also a \kPDS\ of $G-e$, then we are done, so assume $T$ is not. Let $j_0$ be the smallest integer $j$ such that $\Pt {G,k} j (T) \supsetneq \Pt {G-e,k} j (T)$, and let $v$ be a vertex in $\Pt {G,k} {j_0} (T) \setminus \Pt {G-e,k} {j_0} (T)$. Since $v \in \Pt {G,k} {j_0} (T)$, there exists some neighbour $u$ of $v$ in $\Pt {G,k} {j_0-1} (T)$ such that $|N_{G}[u]\setminus \Pt {G,k} {j_0-1} (T)| \le k$. Since $N_{G-e}[u]\subseteq N_G[u]$, $N_{G-e}[u]$ is also included in $\Pt {G-e,k} {j_0} (T)$, and $v$ cannot be a neighbour of $u$ any more, so $e=uv$. Thus we choose $S=T\cup \{v\}$ and using Lemma~\ref{l:general} (with the obvious mapping from $G-e$ to $G$, and $i=j=j_{0}$), we get that $S$ is a \kPDS\ of $G-e$ of order $\gpk(G)+1$. We also get that if $\gpk(G-e)= \gpk(G)+1$, then $\grad(G-e)\le \grad(G)$.

We now prove that $\gpk(G)-1\le \gpk(G-e)$. Let $T$ be a  minimum \kPDS\ of $H=G-e$ and $u$ be an end vertex of $e$. We apply Lemma~\ref{l:general}, for $S=T\cup \{u\}$ and  $i=j=0$. We get that $S$ is a \kPDS\ of $G$ and $\grad(G,S)= \grad (G-e,T)$. We infer that if $S$ is minimal (that is $\gpk(G)=\gpk(G-e)+1$), then  $\grad(G)\le \grad(G-e)$.
\end{proof}

As a first illustration of these possibilities, in the graph $G$ drawn in Fig.~\ref{G-e}, the removal of the edge $e_1$ decreases the $k$-power domination number, the removal of the edge $e_3$ increases it, and the removal of the edge $e_2$ does not have any consequence.
\begin{figure}[h]
 \begin{center}
 \begin{tikzpicture}[thick,scale=0.8,
	vertex/.style={circle,draw,inner sep=0pt,minimum size=1.5mm,fill=white!100}]


  \coordinate (c) at (6,-5);
  
	\foreach \i/\j in {1/2,2/4,4/7}{ 
 	  \coordinate (wp\i) at (0,-1-\j);
	  \coordinate (w\i) at (2,-1-\j);
	  \coordinate (u\i) at (3.5+0.5*\j,-8);
	  }
	\foreach \i in {1,2,3}{ 
	  \path (c) ++(4,0) +(-72*\i-72:2) coordinate (v\i);	}
	  \path (c) ++(4,0) +(-72*5-72:2) coordinate (v4);

	 
	\foreach \i in {1,2,4}{
		\foreach \j in {1,2,4}{
			\draw (wp\i) -- (w\j);
  	}
		\foreach \j in {1,2,3,4}{
			\draw (v\i) -- (v\j);
		}
  	\draw (w\i) -- (c);
  	\draw (u\i) -- (c);
  	
  }
  \draw (v1) -- (c);
  \draw (v2) -- (c);  

	\draw[dashed, thin]  (v3) to[bend left=72] (v4);  
  
  
	\draw (c) node[vertex,label=above:{$u$}] {};
 	\draw (v3) node[vertex] {}; 
 	
  \foreach \i in {1,2}{ 
  	\draw (wp\i) node[vertex,label=left:{$w'_\i$}] {};  
	  \draw (w\i) node[vertex,label=30:{$w_\i$}] {};
		\draw (u\i) node[vertex,label=below:{$u_\i$}] {};
		\draw (v\i) node[vertex,label=(-72*\i-72):{$v_\i$}] {};  
	}
  \def\i{4};
  {\draw (wp\i) node[vertex,label=left:{$w'_{k+1}$}] {};  
	   \draw (w\i) node[vertex,label=-30:{$w_{k+1}$}] {};
		\draw (u\i) node[vertex,label=below:{$u_{k+1}$}] {};
		\draw (v\i) node[vertex,label=(-72*\i-144):{$v_{k+2}$}] {};  
		}
  \def\i{3};
		\draw (v\i) node[vertex,label=(-72*\i-72):{$v_\i$}] {};  
		

	\path (wp1) -- (w1)	node[midway,above] {$e_1$};
	\path (w1) -- (c)	node[midway,above] {$e_2$};
	\path (c) -- (v2)	node[midway,above] {$e_3$};

	\draw (0,-6.5) node[rotate=90] {$\ldots$};  
	\draw (2,-6.5) node[rotate=90] {$\ldots$};  
	\draw (6.25,-8) node {$\ldots$};  
	\draw (c) ++(4,0) ++(-72*4-72:1.5) node[rotate=90] {$\ldots$};

 \end{tikzpicture}

 \caption{A graph $G$ where $\gpk(G)= 2 = \gpk(G-e_{2}), \gpk(G-e_{1})= 1, \gpk(G-e_{3})= 3$.}
 \label{G-e}
 \end{center}
\end{figure}
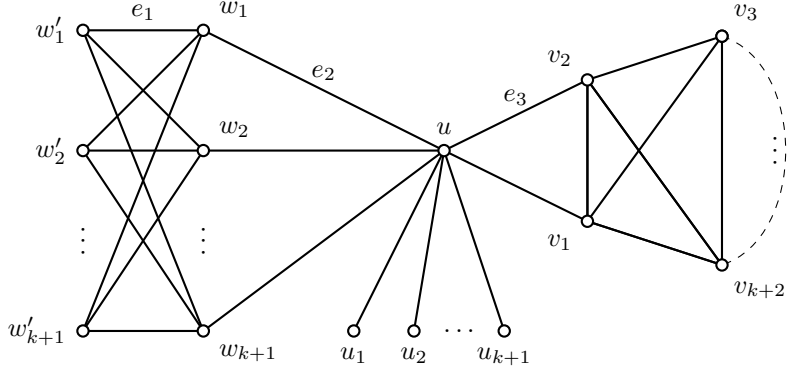

We now propose a graph family where the removal of an edge decreases the $k$-power domination number but increases its propagation radius arbitrarily. The graph $G_{k,r,a}$ represented in Fig.~\ref{fig.G-einc} satisfies $\gpk(G)=2$ and $\grad(G)=a+2$ (which is reached with the initial set $\{u,v\}$). If the edge $e$ is removed, we get a new graph whose $k$-power domination number is $1$ and which has propagation radius $(r+3)(a+1)+2$. So no upper bound can be found for $\grad(G-e)$ (in terms of $\grad(G)$) when the removal of an edge decreases the power domination number.

Similar graphs where the edge removal increases the power domination number can also be found.
For example, in the graph $G_{k,r,a}$, if we remove the topmost path of length $a+2$ from $w$ to $v$, except for the vertex adjacent to $v$, we get another graph $G'$ such that $\{u\}$ is the only  $\gpk(G')$-set of order 1, and with $\grad(G')=(r+2)(a+1)+3$. Removing the same edge $e$, now $\{u,v\}$ is a minimum $\gpk(G'-e)$-set and $\grad(G'-e)=a+2$. This illustrates the fact that no lower bound can be found for $\grad(G-e)$ (in terms of $\grad(G)$) when the removal of an edge increases the power domination number.

\begin{figure}
\begin{center}
\begin{tikzpicture}[thick,scale=0.5,
	vertex/.style={circle,draw,inner sep=0pt,minimum size=1.5mm,fill=white!100}]

\def\k {3}; 
\def\kp {4}; 
\def\km {2}; 
\def\r {4}; 
\def\a {4}; 


\coordinate (c) at (0,0);

\foreach \i in {1,...,\kp}{
		\path (110+140*\i/\kp-70/\kp:2) coordinate (l\i);
		\draw (c)--(l\i);
		\draw (l\i) node[vertex] {};
		}

\draw (0,0) -- (2,0) coordinate (a1) -- (3,0) coordinate (a2) +(\a,0) coordinate (a3);
\draw[decorate,decoration={zigzag,amplitude=1,segment length=4}] (a2) -- (a3) node[midway,above] {$a$};

\path (7+2*\a,0) coordinate (cp);

\foreach \i in {1,...,\kp}{
	  \path (a3) -- ++(2,-\kp/2-1/2+\i) coordinate (b\i) -- +(\a,0) coordinate (c\i);
	  \draw (a3) -- (b\i) (c\i) -- (cp);
	  \draw[decorate,decoration={zigzag,amplitude=1,segment length=4}] (b\i) -- (c\i);
	  
	  }
\path (a3) node[above] {$w$} -- (b1) node[midway, below] {$e$};
\path (b2) -- (c2) node[midway,above] {$a$};

\foreach \j in {1,...,\r}{
	  \path (2,-\kp/2*\r+\kp*\j-\kp/2) coordinate (d\j);
	  \path (10+2*\a,-\kp/2*\r+\kp*\j-\kp/2) coordinate (e\j);
	  \draw (c) -- (d\j)--(e\j) -- (cp);
		\foreach \i in {1,...,\k}{
			\path (e\j) -- ++(2,-\kp/2+\i) coordinate (f\j\i) -- ++(\a,0) coordinate (g\j\i);
			\draw (e\j) -- (f\j\i);
			\draw[decorate,decoration={zigzag,amplitude=1,segment length=4}] (f\j\i) -- (g\j\i);
			}	  
	  }
\foreach \i in {1,...,\km}{
	  \draw (cp) -- ++(-80-40*\i/\km:1.8) coordinate (cp\i);
	  \draw (cp\i) node[vertex] {};

	  }

\draw[dotted] (cp) ++(-90:1.5) arc(-90:-135:1.5) +(0.5,-1.5) node {$k-1$} ;

\path (f\r\k) -- (g\r\k) node[midway,above] {$a$};

\draw 
	(e1) [bend left=20] to (g21)
	(e2) [bend left=20] to (g31)
	(e3) [bend left=20] to (g41);

\draw 
	(cp) [bend right=20] to (g1\k)
	(cp) [bend left=5] to (g21)
	(cp) [bend right=5] to (g31)
	(cp) [bend left=20] to (g41);	

\foreach \i in {1,...,\kp}{
    \draw (b\i) node[vertex] {};
    \draw (c\i) node[vertex] {};
}

\foreach \j in {1,...,\r}{
		\draw (d\j) node[vertex] {};
		\draw (e\j) node[vertex] {};
		\foreach \i in {1,...,\k}{
			\draw (f\j\i) node[vertex] {};
			\draw (g\j\i) node[vertex] {};
		}
}

\draw (a1) node[vertex] {};
\draw (a2) node[vertex] {};
\draw (a3) node[vertex] {};
\draw (c) node[vertex, label=$u$] {};
\draw (cp) node[vertex, label=$v$] {};

\draw (e4) node[above] {$x_1$};
\draw (e3) node[below] {$x_2$};
\draw (e1) node[below] {$x_r$}; 
 
\draw[dotted] (110:1.5) arc (120:250:1.5) node[midway, left] {$k+1$};
\draw[dotted] (a3) ++(60:1) node[above] {$k+1$} arc (60:-60:1);
\draw[dotted] (0.5,3.5) node[above] {$r+1$} arc (20:-20:10.5);
\draw[dotted] (e\r) ++(60:1) node[above] {$k$} arc (60:-60:1);
 
\end{tikzpicture}

 \caption{The graph $G_{k,r,a}$ for $k=3$ and $r=4$ (zigzag edges represent paths of length $a$).}
  \label{fig.G-einc}
\end{center}
\end{figure}
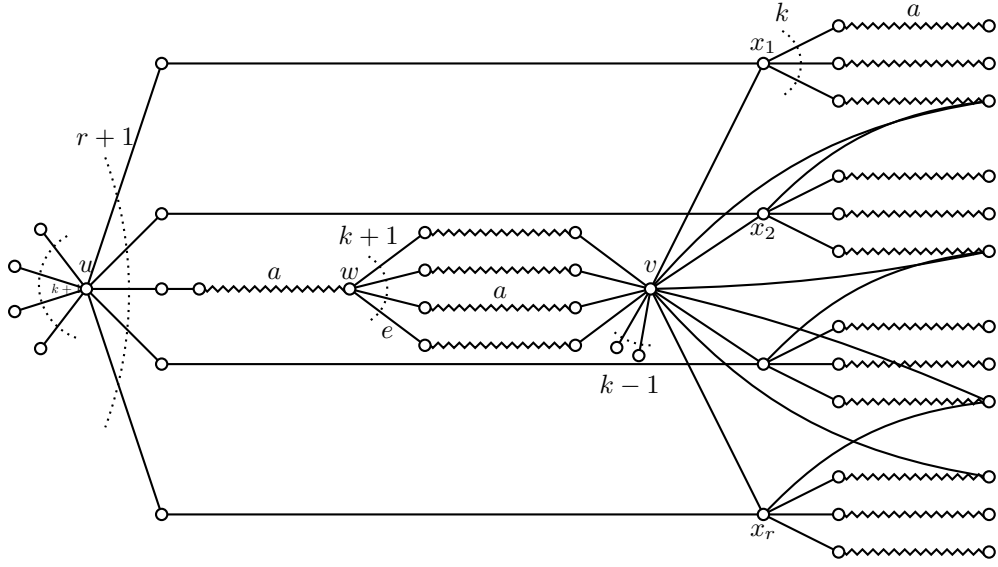

%
%
%
%
%

\medskip

We now characterize the graphs for which the removal of any edge increases the power domination number. Define a {\em $k$-generalized spider} as a tree with at most one vertex of degree $k+2$ or more. See Fig.~\ref{G.e} for an example.

\begin{figure}[h]\center
\begin{tikzpicture}[thick,scale=1,
	vertex/.style={circle,draw,inner sep=0pt,minimum size=1.5mm,fill=white!100}]

	\path (0,0) coordinate (c1) ++(1,0) coordinate (c2)
	    ++(1.5,0) coordinate (d1) ++(1,0) coordinate (d2)  
	    ++(1.5,0) coordinate (f1) ++(1,0) coordinate (f2)
	    ++(1.5,0) coordinate (g1) ++(1,0) coordinate (g2)
	    ;

	\path (0.5,1) coordinate (b1) ++(2.5,0) coordinate (b2)
	    ++(2.5,0) coordinate (e1) ++(2.5,0) coordinate (e2)
	    ;

	\path (1.75,2) coordinate (a1) ++(5,0) coordinate (a2)
	    ;

	\path (4.25,3) coordinate (v);

	
	\foreach \i/\j in {v/a,a1/b,a2/e,b1/c,b2/d,e1/f,e2/g}{
		\foreach \k in {1,2}{
	  	\draw (\i) -- (\j\k);
	  	}
    }
		    
	\foreach \i in {c,d,f,g}{	
		\foreach \j in {1,2}{
 			\draw (\i\j) node[vertex,label=below:{$\i_\j$}] {};
 			}
 		}	
	\foreach \i in {a,b,e}{	
 		\draw (\i1) node[vertex,label=left:{$\i_1$}] {};
 		}	
 	\draw (a2) node[vertex,label=right:{$a_{k+2}$}] {};
 	\draw (b2) node[vertex,label=right:{$b_{k+1}$}] {};
 	\draw (e2) node[vertex,label=right:{$e_{k+1}$}] {};

	\draw (v) node[vertex,label=above:{$v$}] {};
	
	
	\draw (1.75,0) node {$\ldots$} ++(2.5,0) node {$\ldots$} ++(2.5,0) node {$\ldots$};
	\draw (1.75,1) node {$\ldots$} ++(2.5,0) node {$\ldots$} ++(2.5,0) node {$\ldots$};	
	\draw (4.25,2) node {$\ldots$};
	
\end{tikzpicture}

  \caption{A $k$-generalized spider, $T$}
  \label{G.e}
\end{figure}
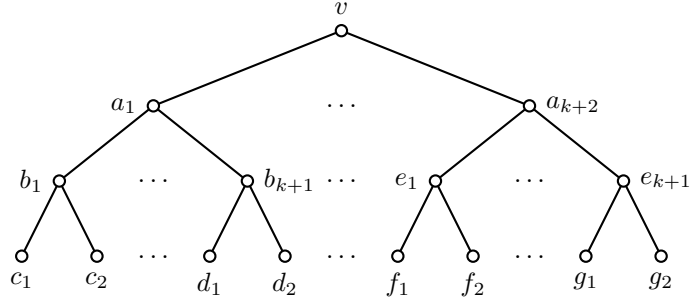

\begin{theorem}
Let $G$ be a graph. For each edge $e$ in $G$, $\gpk(G-e)>\gpk(G)$ if and only if $G$ is a disjoint union of $k$-generalized spiders. 
\end{theorem}

\begin{proof}

First observe that if $G$ is a disjoint union of $k$-generalized spiders, then its $k$-power domination number is exactly its number of components, and clearly $\gpk(G-e)>\gpk(G)$ for any edge $e$ in $G$. 

Let $G$ be a graph and let $S$ be a $\gpk(G)$-set. We label the vertices of $G$ with integers from $1$ to $n$ and consider the subsequent natural ordering on the vertices. For $i\ge 0$, we define $E'_i\subseteq E(G)$ as follows:
\begin{displaymath}\begin{cases}
E'_0=\left\{uv\in E(G) \mid v\in N(S)\setminus S, u=\min\{x \in N(v)\cap S\}\right\} \\
E'_{i+1}=\left\{uv \in E(G) \mid v \in \Pt{k} {i+1} (S)\setminus \Pt{k} {i} (S), u=\min\{x\in \Pt{k} {i} (S) \cap N(v), |N[x]\setminus \Pt {k} {i} (S)|\le k \} \right\} 
\end{cases}
\end{displaymath}
where the minima are taken according to the ordering of the vertices. Let $E'$ be the union of all $E'_i$ for $i\ge 0$. If we consider the edges of $E'$ as defined above oriented from $u$ to $v$, then the in-degree of each vertex not in $S$ is 1, of vertices in $S$ is 0. Also the graph is acyclic, and each vertex not in $S$ has out-degree at most $k$. Thus the graph induced by $E'$ is a forest of $k$-generalized spiders. Note also that $S$ is a \kPDS\ of this graph. We now assume that for any edge $e \in E(G)$, $\gpk(G-e)>\gpk(G)$, and we then prove that $E'=E(G)$. 

By way of contradiction, suppose there exists an edge $e$ in $E(G)$ and not in $E'$. We prove that $S$ is a \kPDS\ of $G-e$. For that, we prove by induction that for all $i, \Pt {G,k} {i} (S)\subseteq \Pt {G-e,k} {i} (S)$. First observe that $\Pt {G-e,k} {0} (S)=\Pt {G,k} {0} (S)$. 
Indeed, suppose there exists a vertex $x$ in $\Pt {G,k} {0} (S)$ but not in $\Pt {G-e,k} {0} (S)$, then $e$ has to be of the form $xv$ with $v\in S$. But since $e\notin E'_0$, there exists another vertex $u<v$ in $S$ such that $ux\in E'_0$, and $x\in \Pt {G-e,k} {0} (S)$.

Assume now $\Pt {G,k} {i} (S)\subseteq \Pt {G-e,k} {i} (S)$ for some $i\ge 0$, and let us prove that $\Pt {G,k} {i+1} (S)\subseteq \Pt {G-e,k} {i+1} (S)$. Let $x$ be a vertex in $\Pt {G,k} {i+1} (S)$. If $x\in \Pt {G,k} {i} (S)$, then by induction hypothesis, $x\in \Pt {G-e,k} {i+1} (S)$. If $x\notin \Pt {G,k} {i} (S)$, then there exists a vertex $v\in \Pt {G,k} {i} (S), x\in N_{G}[v]$ such that $|N_{G}[v]\setminus \Pt {G,k} {i} (S)|\le k$. Suppose $e\neq xv$. Then, since $N_{G-e}[v]\subseteq N_{G}[v]$ and by induction hypothesis, $v\in \Pt {G-e,k} {i} (S), x\in N_{G-e}[v]$ and $|N_{G-e}[v]\setminus \Pt {G-e,k} {i} (S)|\le k$, which implies $x\in \Pt {G-e,k} {i+1} (S)$. If $e=xv$ then by the choice of $E_{i+1}^{'}$, there exists a vertex $w\in \Pt {G,k} {i} (S), w<v, wx\in E_{i+1}^{'}$ such that $|N_{G}[w]\setminus \Pt {G,k} {i} (S)|\le k$ and $x\in N_{G}[w]\setminus \Pt {G,k} {i} (S)$. Then by induction hypothesis, $w\in \Pt {G-e,k} {i} (S), x\in N_{G-e}[w]$ and $|N_{G-e}[w]\setminus \Pt {G-e,k} {i} (S)|\le k$, which implies $x\in \Pt {G-e,k} {i+1} (S)$. Therefore $E(G)=E'$ and $G$ is indeed a union of $k$-generalized spiders.
\end{proof}

Observe that there also exist graphs for which the removal of any edge decreases the power domination number, though we did not manage to characterize them. The simplest example is the complete bipartite graph $K_{k+2,k+2}$, in which the removal of any edge decreases the $k$-power domination number from 2 to 1. This graph already played a noticeable role among the $k+2$-regular graphs, as observed in \cite{dohelomora-13}. Another example is the graph $K_{k+3,k+3}-M$, where $M$ is a perfect matching, in which we have $\gpk(K_{k+3,k+3}-M)=2$ and $\gpk((K_{k+3,k+3}-M)-e)=1$ for any edge $e$. More complex examples are the Cartesian product of  $K_4$ and $ W_5$, where the $k$-power domination number decreases from 3 to 2. A general family of graphs having this property is the Cartesian product of two complete graphs of the same order $K_a\cp K_a$, which shall be described in Section~\ref{s:hamming}.

\subsection{Edge contraction}

Contracting an edge in a graph may decrease its domination number by one, but cannot increase it (see \cite{huxu-12}). As we prove in the following, increasing of the power domination number may occur.

\begin{theorem}\label{t:propG/e}
Let $G$ be a graph and $e$ be an edge in $G$. Then 
\begin{displaymath}
\gpk(G)-1\le \gpk(G/e)\le \gpk(G)+1\,.\end{displaymath}
Moreover,
\begin{displaymath}
\begin{cases}
\mbox{if }\gpk(G)-1 = \gpk(G/e)\mbox{, then } \grad(G)\le \grad(G/e)\\
\mbox{if }\gpk(G/e) = \gpk(G)+1\mbox{, then } \grad(G/e)\le \grad(G).
\end{cases}
\end{displaymath}
\end{theorem}
\begin{proof}
Let $e=xy$ be an arbitrary edge in $G$, we denote by $v_{xy}$ the vertex obtained by contraction of $e$ in $G/e$. We first prove that $\gpk(G/e)\ge \gpk(G)-1$. Let $T$ be a minimum \kPDS\ of $H=G/e$. Suppose first that the vertex $v_{xy}\in T$, then taking $S=T\setminus\{v_{xy}\}\cup \{x,y\}$, the conditions of Lemma~\ref{l:general} hold  from $i=j=0$ with the natural mapping from $G\setminus\{x,y\}$ to $H\setminus v_{xy}$. We infer that $S$ is a \kPDS\ of $G$ and $\grad(G,S)= \grad (G/e,T)$. 
We now consider the case when $v_{xy}\notin T$.
Let $j_{0}$ be the smallest $j$ such that $v_{xy}\in \Pt {G/e,k} j (T)$. Let $w$ be a neighbour of $v_{xy}$ that brought $v_{xy}$ into $\Pt {G/e,k} j (T)$, \textit{i.e.} if $j_0=0$, $w$ is a neighbour of $v_{xy}$ in $T$, otherwise when $j_0>0$,  $w$ is a neighbour of $v_{xy}$ in $\Pt {G/e,k} {j_0-1}(T)$ such that $|N_{G/e}[w]\setminus \Pt {G/e,k} {j_0-1} (T)| \le k$. By definition of edge contraction, the edge $wv_{xy}$ corresponds to an edge $wx$ or $wy$ in $E(G)$. If $wx\in E(G)$, then take $S = T\cup \{y\}$, otherwise take $S = T\cup \{x\}$. Then, applying Lemma~\ref{l:general} (with the natural mapping from $G\setminus\{x,y\}$ to $H\setminus v_{xy}$ and $i=j=j_0$), we get that $S$ is a \kPDS\ of $G$ and $\grad(G,S)= \grad (G/e,T)$. This implies that if $\gpk(G)= \gpk(G/e)+1$, then $\grad(G)\le \grad(G/e)$.  

We now prove that $\gpk(G/e)\le \gpk(G)+1$. Let $T$ be a minimum \kPDS\ of $G$ and let $S=T\setminus \{x,y\}\cup \{v_{xy}\}$. 
Let $j_{0}$ be the smallest $j$ such that $N_{G}[x]\cup N_{G}[y]\subseteq \Pt {G,k} j (T)$.
Here also, we can use Lemma~\ref{l:general} (with the natural mapping from $(G/e)\setminus v_{xy}$ to $G\setminus \{x,y\}$ and $i=j=j_0$), and get that $S$ is \kPDS\ of $G/e$. We also get that if $\gpk(G/e)= \gpk(G)+1$, then $\grad(G/e)\le \grad(G)$.    
\end{proof}

The bounds in Theorem~\ref{t:propG/e} are tight. For example, the lower bound holds for the graphs $K_{k+2,k+2}$ and $K_{k+3,k+3}-M$, where $M$ is a perfect matching, but also for the Cartesian product of two complete graphs of same order $K_a\cp K_a$, as is described in the next section. The upper bound is attained for example for the $k$-generalized spider $T$ in Fig.~\ref{G.e}, which satisfy $\gpk(T)=1$ and $\gpk(T/a_{1}b_{1})=2$ for $k\ge 2$.

\subsection{On the Cartesian product of twin complete graphs}\label{s:hamming}

The Cartesian product of two complete graphs of same (large enough) order is such that removing a vertex, removing an edge or contracting an edge decrease its power domination number. We here prove these properties.

\begin{observation}\label{t:gpkKa}
Let $a\ge 1$ and $G=K_a\cp K_a$. Then $\gpk(G)=\begin {cases}a-k& \textrm{if } a\ge k+2\,,\\
                       1 & \textrm{otherwise.}\end{cases}$   
\end{observation}

\begin{proof}
Denote by $\{v_1, \hdots , v_a\}$ the vertices of $K_a$. If $a<k+2$, then any vertex in $G=K_a\cp K_a$ is a minimum \kPDS. Now, assume $a\ge k+2$. Let $S=\{(v_i, v_i)\mid 1\le i\le a-k\}$. Then $\Pt k 0 (S)=\{(v_i,v_j)\mid i\le a-k \mbox{ or } j\le a-k\}$ and  the set of vertices $A=\{(v_i, v_j)\mid a-k+1\le i,j\le a\}$ is yet to be monitored. Since any vertex in $\Pt k 0 (S)\setminus A$ has either 0 or $k$ neighbours in $A$ and each vertex in $A$ is adjacent to some vertex in $\Pt k 0 (S)$, $\Pt k 1 (S)$ covers the whole graph. Thus $S$ is a \kPDS\ of $G$. Therefore, $\gpk(G)\le a-k$.

We now prove that $\gpk(G)\ge a-k$. By way of contradiction, suppose $S$ is a \kPDS\ of $G$ such that $|S|\le a-k-1$. Without loss of generality, assume that the elements of $S$ belong to the first $a-k-1$ columns and rows of $G$. Then the vertices in the set $B=\{(v_i, v_j)\mid a-k\le i,j\le a\}$ are not adjacent to any vertex in $S$, and $\Pt k 0 (S)\cap B=\emptyset$. Since any vertex in $G\setminus B$ has either 0 or $k+1$ neighbours in $B$, no vertices from this set may get monitored later on, a contradiction.     
\end{proof}

\begin{observation}
Let $a\ge k+2$ and $G=K_a\cp K_a$. Then $\gpk(G-v)=a-k-1$ for any vertex $v$ in $V(G)$.
\end{observation} 

\begin{proof}
Denote by $\{v_1, \hdots , v_a\}$ the vertices of $K_a$. We prove the result for $v=(v_1, v_1)$ which implies the result for any $v$ by vertex transitivity. First observe that $S=\{(v_i, v_i) \mid 2\le i\le a-k\}$ is a \kPDS\ of $G-v$. Indeed $\Pt k 0 (S)=\{(v_i,v_j)\mid  2\le i\le a-k \mbox{ or }2\le j\le a-k\}$ then vertices $(v_i,v_1)$ (resp. $(v_1,v_i)$)  with $2\le i\le a-k$ have only vertices $(v_j,v_1)$ (resp. $(v_1,v_j)$) with $a-k+1\le j\le a$ as unmonitored neighbours, which are thus all in $\Pt k 1 (S)$. The next propagation step covers the graph. Thus $S$ is a \kPDS\ of $G-v$ and $\gpk(G-v)\le a-k-1$. Now by Theorem~\ref{o:propG-v} and Observation~\ref{t:gpkKa}, $\gpk(G-v)\ge a-k-1$.  
\end{proof}

\begin{observation}
Let $a\ge k+2$ and $G=K_a\cp K_a$. Then $\gpk(G-e)=a-k-1$ for any edge $e$ in $E(G)$.
\end{observation}

\begin{proof}
Denote by $\{v_1, \hdots , v_a\}$ the vertices of $K_a$. By edge transitivity of $G$, we can assume that $e=(v_1, v_1)(v_2, v_1)$. Let $S=\{(v_i, v_i)\mid 2\le i\le a-k\}.$ Then $\Pt k 0 (S)=\{(v_i,v_j)\mid  2\le i\le a-k \mbox{ or }2\le j\le a-k\}$. Now the vertex $(v_2, v_1)$ has only $k$ unmonitored neighbours, namely the vertices $(v_j,v_1)$ for $a-k < j\le a$, and they all are in $\Pt k 1 (S)$. 
Then all vertices $(v_j,v_2)$ for $a-k < j\le a$ have only $k$ unmonitored neighbours and thus $\Pt k 2 (S)$ contains all vertices $(v_i,v_j)$ for $i\ge 2$. Then $\Pt k 3 (S)$ contains the whole graph and $\gpk(G-e)\le a-k-1$. The lower bound follows from Theorem~\ref{t:propG-e} and Observation~\ref{t:gpkKa}.
\end{proof} 

\begin{observation}
Let $a\ge k+2$ and $G=K_a\cp K_a$. Then $\gpk(G/e)=a-k-1$ for any edge $e$ in $E(G)$.
\end{observation}

\begin{proof}
Denote by $\{v_1, \hdots , v_a\}$ the vertices of $K_a$. By edge transitivity of $G$, we can assume that $e=(v_1, v_1)(v_2, v_1)$ and we denote by $v_e$ the vertex in $G/e$ obtained by contracting $(v_1, v_1)$ and $(v_2, v_1)$. Let $S=\{v_e\} \cup \{(v_i, v_i)\mid 3\le i\le a-k\}$. Then $\Pt k 0 (S)$ contains all vertices $(v_i,v_j)$ with $1\le i\le a-k$ and $1\le j\le a$. After one propagation step, the whole graph is monitored so $\gpk(G/e)\le a-k-1$. The lower bound follows from Theorem~\ref{t:propG/e} and Observation~\ref{t:gpkKa}.  
\end{proof} 

\acknowledgements
The authors thank the Erudite programme of the Kerala State Higher Education Council, Government of Kerala, India for funding the visit of the first author during March 2014. The second author is supported by Maulana Azad National Fellowship (F1-17.1/2012-13/MANF-2012-13-CHR-KER-15793) of the University Grants Commission, India.

\nocite{*}
\bibliographystyle{abbrvnat}
\bibliography{dosevi}

\begin{thebibliography}{17}
\providecommand{\natexlab}[1]{#1}
\providecommand{\url}[1]{\texttt{#1}}
\expandafter\ifx\csname urlstyle\endcsname\relax
  \providecommand{\doi}[1]{doi: #1}\else
  \providecommand{\doi}{doi: \begingroup \urlstyle{rm}\Url}\fi

\bibitem[Aazami(2010)]{aaz-10}
A.~Aazami.
\newblock Domination in graphs with bounded propagation: algorithms,
  formulations and hardness results.
\newblock \emph{J. Comb. Optim.}, 19\penalty0 (4):\penalty0 429--456, 2010.

\bibitem[Aazami and Stilp(2009)]{aast-09}
A.~Aazami and K.~Stilp.
\newblock Approximation algorithms and hardness for domination with
  propagation.
\newblock \emph{SIAM J. Discrete Math.}, 23\penalty0 (3):\penalty0 1382--1399,
  2009.

\bibitem[Baldwin et~al.(1993)Baldwin, Mili, Boisen~Jr, and Adapa]{bamiboad-93}
T.~Baldwin, L.~Mili, M.~Boisen~Jr, and R.~Adapa.
\newblock Power system observability with minimal phasor measurement placement.
\newblock \emph{IEEE Trans. Power Systems}, 8\penalty0 (2):\penalty0 707--715,
  1993.

\bibitem[Chang et~al.(2012)Chang, Dorbec, Montassier, and Raspaud]{chdomora-12}
G.~J. Chang, P.~Dorbec, M.~Montassier, and A.~Raspaud.
\newblock Generalized power domination of graphs.
\newblock \emph{Discrete Appl. Math.}, 160\penalty0 (12):\penalty0 1691--1698,
  2012.

\bibitem[Dorbec and Klav{\v{z}}ar(2014)]{dokl-14}
P.~Dorbec and S.~Klav{\v{z}}ar.
\newblock Generalized power domination: propagation radius and {S}ierpi\'nski
  graphs.
\newblock \emph{Acta Appl. Math.}, 134:\penalty0 75--86, 2014.

\bibitem[Dorbec et~al.(2008)Dorbec, Mollard, Klav{\v{z}}ar, and
  {\v{S}}pacapan]{domoklsp-08}
P.~Dorbec, M.~Mollard, S.~Klav{\v{z}}ar, and S.~{\v{S}}pacapan.
\newblock Power domination in product graphs.
\newblock \emph{SIAM J. Discrete Math.}, 22\penalty0 (2):\penalty0 554--567,
  2008.

\bibitem[Dorbec et~al.(2013)Dorbec, Henning, L{\"o}wenstein, Montassier, and
  Raspaud]{dohelomora-13}
P.~Dorbec, M.~A. Henning, C.~L{\"o}wenstein, M.~Montassier, and A.~Raspaud.
\newblock Generalized power domination in regular graphs.
\newblock \emph{SIAM J. Discrete Math.}, 27\penalty0 (3):\penalty0 1559--1574,
  2013.

\bibitem[Dorfling and Henning(2006)]{dohe-06}
M.~Dorfling and M.~A. Henning.
\newblock A note on power domination in grid graphs.
\newblock \emph{Discrete Appl. Math.}, 154\penalty0 (6):\penalty0 1023--1027,
  2006.

\bibitem[Ferrero et~al.(2011)Ferrero, Varghese, and Vijayakumar]{FSV-11}
D.~Ferrero, S.~Varghese, and A.~Vijayakumar.
\newblock Power domination in honeycomb networks.
\newblock \emph{J. Discrete Math. Sci. Cryptogr.}, 14\penalty0 (6):\penalty0
  521--529, 2011.

\bibitem[Guo et~al.(2008)Guo, Niedermeier, and Raible]{gunira-08}
J.~Guo, R.~Niedermeier, and D.~Raible.
\newblock Improved algorithms and complexity results for power domination in
  graphs.
\newblock \emph{Algorithmica}, 52\penalty0 (2):\penalty0 177--202, 2008.

\bibitem[Haynes et~al.(1998{\natexlab{a}})Haynes, Hedetniemi, and
  Slater]{hahesl-98}
T.~W. Haynes, S.~T. Hedetniemi, and P.~J. Slater.
\newblock \emph{Fundamentals of domination in graphs}, volume 208 of
  \emph{Monographs and Textbooks in Pure and Applied Mathematics}.
\newblock Marcel Dekker, Inc., New York, 1998{\natexlab{a}}.

\bibitem[Haynes et~al.(1998{\natexlab{b}})Haynes, Hedetniemi, and
  Slater]{hahesl-98b}
T.~W. Haynes, S.~T. Hedetniemi, and P.~J. Slater.
\newblock \emph{Domination in Graphs: Advanced Topics}.
\newblock Monographs and Textbooks in Pure and Applied Mathematics. Marcel
  Dekker, Inc., New York, 1998{\natexlab{b}}.

\bibitem[Haynes et~al.(2002)Haynes, Hedetniemi, Hedetniemi, and
  Henning]{hahehehe-02}
T.~W. Haynes, S.~M. Hedetniemi, S.~T. Hedetniemi, and M.~A. Henning.
\newblock Domination in graphs applied to electric power networks.
\newblock \emph{SIAM J. Discrete Math.}, 15\penalty0 (4):\penalty0 519--529,
  2002.

\bibitem[Huang and Xu(2012)]{huxu-12}
J.~Huang and J.-M. Xu.
\newblock Note on conjectures of bondage numbers of planar graphs.
\newblock \emph{Appl. Math. Sci. (Ruse)}, 6\penalty0 (65-68):\penalty0
  3277--3287, 2012.

\bibitem[Liao(2016)]{lich-14}
C.-S. Liao.
\newblock Power domination with bounded time constraints.
\newblock \emph{J. Comb. Optim.}, 31\penalty0 (2):\penalty0 725--742, 2016.

\bibitem[Wang et~al.(2016)Wang, Chen, and Lu]{walech-14}
C.~Wang, L.~Chen, and C.~Lu.
\newblock {$k$}-{P}ower domination in block graphs.
\newblock \emph{J. Comb. Optim.}, 31\penalty0 (2):\penalty0 865--873, 2016.

\bibitem[Zhao et~al.(2006)Zhao, Kang, and Chang]{zhkach-06}
M.~Zhao, L.~Kang, and G.~J. Chang.
\newblock Power domination in graphs.
\newblock \emph{Discrete Math.}, 306\penalty0 (15):\penalty0 1812--1816, 2006.

\end{thebibliography}
\label{sec:biblio}

\end{document}